\documentclass[11pt]{article}

\usepackage{graphicx}
\usepackage{enumitem}
\usepackage{amsmath}
\usepackage{amssymb}
\usepackage{amsthm}
\usepackage{hyperref}
\usepackage[letterpaper,margin=1in]{geometry}

\newtheorem{theorem}{Theorem}
\newtheorem{lemma}{Lemma}

\title{Designing Local Distributed Mechanisms}
\author{Juho Hirvonen \\ Helsinki Institute for Information Technology HIIT and Aalto University, Finland \\
\texttt{juho.hirvonen@aalto.fi} \and Sara Ranjbaran \\ Aalto University, Finland \\ \texttt{sara.ranjbaran@aalto.fi}}
\date{June 26, 2026}

\begin{document}

\maketitle

\begin{abstract}
	Mechanism design studies algorithms as games with the goal of designing algorithms that incentivise participants to be truthful: to behave according to their real preferences. In this work we introduce a new notion: local distributed mechanisms. These are truthful mechanisms that have an implementation as fast distributed algorithms and have non-trivial approximation guarantees. We show how monotone distributed optimisation algorithms can be automatically transformed into truthful mechanisms using the seminal result known as Myerson's Lemma. We demonstrate mechanisms for four fundamental graph problems: maximum-weight independent set, minimum-weight vertex cover, minimum-weight dominating set, and a variant of weighted colouring.
	
	We study distributed mechanism design in two settings. In the trusted setting we are only interested in a distributed implementation of a mechanism and assume that the algorithm is executed correctly. We also study what happens when the strategic agents themselves are responsible for running the algorithm, and show that a large class of algorithms cannot be implemented. Despite this impossibility, we are able to design a non-trivial incentive-compatible mechanism for weighted independent set problem.
	
	Our work extends previous work in Distributed Algorithmic Mechanism Design (DAMD) in a new direction. Instead of studying global problems like routing or leader election, we study local resource allocation problems. Our algorithms are simple and thus potentially practical. Local algorithms are particularly interesting for highly dynamic large-scale systems, and there are many potential future application domains, e.g.\ demand-side load management in electric grids or resource allocation in IoT computing.
\end{abstract}

\section{Introduction}

In mechanism design the goal is to implement a social choice function as a game. Participating agents hold private information that is critical to choosing a good assignment. The task of the mechanism designer is to come up with an algorithm and incentives such that the agents should always reveal their private information to the mechanism. A classic example is the sale of a single item: how should we assign the item and what should it cost? One of the most common solutions is a sealed-bid first-price auction: each participant sends a bid in secret and the largest bidder gets the item, paying its own bid. The problem with this mechanism is that it is game-theoretically unstable: the optimal behaviour for the winner is to bid just above the (unknown) second-highest bid. Trying to reason about the valuations of the other agents makes participating in this auction difficult. The \emph{second-price} auction fixes this issue: now the winner pays the second highest bid, removing the need to guess what the second bid was~\cite{vickrey61auction}. This can be seen as the auction automatically applying the optimal strategy of the first-price auction for the winner. The second-price auction is \emph{incentive-compatible}: bidding their true valuation guarantees the optimal outcome for all agents.

Mechanism design has many important practical applications in addition to auctions. The deferred acceptance algorithm can be used to e.g.\ assign students to schools~\cite{gale62sm,nobel12} and the top trading cycles algorithm to create stable chains of organ donations~\cite{roth04kidney}. Google uses large-scale automated auctions to sell the advertisements on its searches~\cite{edelman07internet}.

Mechanism design typically assumes a single entity that runs the underlying algorithm and with which all participating agents must interact. In this work we study how this bottleneck can be removed: we want to design mechanisms that \emph{can be implemented as fast distributed algorithms}. Mechanism design for distributed systems has been studied previously (\emph{distributed algorithmic mechanism design (DAMD)})~\cite{feigenbaum07book}. To our knowledge, all previous work in this area has been on fundamentally global problems, such as routing and leader election. In contrast, we want to study fundamentally \emph{local} graph problems, where the goal is to solve locally constrained graph problems, such as scheduling conflicting transmissions or ensuring spatial coverage. We believe understanding such mechanisms will be an important building block in future automated systems: as examples, mechanisms have been proposed for resource allocation in IoT networks~\cite{Prabodini-IoT} and simple mechanisms are already used to change consumer behaviour in electric grids, trying to match electricity demand and generation~\cite{jordehi19optimisation}.

We introduce the notion of \emph{local distributed mechanisms}. These are mechanisms for resource allocation in networks that can be implemented as fast distributed algorithms. As an example, consider the problem of computing an independent set of large weight in the mechanism design setting: each agent has a private parameter, its weight, that gives its payoff for being part of the independent set. The goal is to, essentially, arrange a distributed auction to decide which agents should join the independent set. We study distributed algorithms in the standard LOCAL and CONGEST models~\cite{Peleg2000}.

Local distributed mechanisms can be designed based on the seminal result in mechanism design known as Myerson's lemma~\cite{myerson81optimal}. This result states that in \emph{single-parameter} settings such as the weighted independent set problem, an incentive-compatible mechanism exists exactly when the allocation rule is monotone: increasing an agent's bid will never cause it to no longer be selected. Crucially the result does not depend in any way on the algorithm that implements the allocation rule: we can use a monotone \emph{distributed} optimisation algorithm. Myerson's lemma also gives the formula for the payments, known as the \emph{critical price}. In the context of weighted independent set, each selected agent has to pay the smallest bid that results in it being selected. This is a generalisation of the second-price auction~\cite{vickrey61auction}: the mechanism simulates optimal bids in a first-price auction.

In mechanism design a key assumption is the existence of a central entity known as the \emph{principal}: all agents communicate with the principal which is responsible for computing the outcome. The principal is also responsible for making or collecting the payments of the agents. While communicating the private inputs to a centralised entity is not necessary in the distributed setting, we will continue assuming the existence of an outside entity that handles the payments. The existence of such an outside entity is necessary to an extent, as the payments generally need to be removed from the system (instead of just be redistributed) to the required effect on incentives. We believe this model is reasonable in highly dynamic settings: the structure and valuations change constantly, and it is important that the mechanism is able to change its allocation to match these changes. On the other hand the payments can be recorded locally and be handled eventually by the principal. 

Another important modelling choice is the role of the agents in running the mechanism itself. We study two variants of what we call the \emph{strategic LOCAL model}: in the \emph{trusted} model, the agents only reveal their own inputs and the mechanism is run by a trusted distributed layer. In the \emph{consistent} model the strategic agents themselves are responsible for computation and communication, but must report their outputs consistently to the principal. The trusted model is motivated by settings where the strategic agents are humans whose actions need to be automatically coordinated based on their preferences. In electric grids, there are many problems related to \emph{demand response}~\cite{jordehi19optimisation} where it would be useful to reduce electricity demand during peak demand or low supply. For example, charging of electric vehicles can exhaust local grid capacity. A distributed mechanism for this problem would take as input the preferences of the human customers and output a schedule for automatic vehicle charging with price adjustments to incentivise the customers to use such an allocation. The role of the human agent is to provide their preferences to a trusted system. 

\subsection{Our contribution}

Our work represents the first systematic study of implementing mechanisms as local distributed graph algorithms. We study two main questions: whether mechanisms can be implemented as distributed algorithms and whether these algorithms can be run by the agents themselves.

We begin by showing that a seminal result mechanism design, Myerson's Lemma, is local. This means that any monotone distributed optimisation algorithm in the LOCAL model can be automatically transformed into an incentive-compatible mechanism with the same distributed time complexity (Theorem~\ref{thm:main}). The key observation, given in Lemma~\ref{lem:local-payments}, is that computing the critical prices (i.e.\ deciding what is the smallest bid that causes an agent to be selected) is a local property: in the LOCAL model each node can gather full information about its neighborhood and simulate the algorithm with every possible bid.

To illustrate this result, we present an incentive-compatible mechanism for the $\Delta$-approximation of maximum-weight independent set in Section~\ref{sec:independent-set}. This mechanism is based on a simple greedy algorithm that repeatedly selects the local maxima of a graph. As it turns out, the simplicity of this algorithm is very important for the implementation of mechanisms in more constrained settings. 

The method for computing the critical prices does not automatically translate to the CONGEST model. In Section~\ref{sec:independent-set} we show that the weighted independent set mechanism can also be implemented in the CONGEST model. The key idea is to observe that there is a much simpler combinatorial property of the solution that determines the critical price: If a selected agent were to bid under the bid of its neighbor, then that neighbor would join the set if and only if it did not have another selected neighbor with a higher bid blocking it.

In Section~\ref{sec:unreliable} we show that the mechanism for weighted independent set can also be implemented in a setting where the agents themselves are responsible for the computation and the communication (Theorem~\ref{thm:consistent-is}). In contrast we show that algorithms that are flexible enough and depend on information relayed by neighboring agents (e.g.\ a greedy algorithm with a score function scaled with agent's degree) do not have an incentive-compatible implementation as a mechanism. Our algorithm relies on two key properties. The greedy selection is based on pairwise comparisons, and an agent is selected if it wins all of its neighbors. Therefore there is no need to directly receive information about the network beyond the immediate neighbors. Similarly, to construct an incentive-compatible tie-breaking (coloring) algorithm, we use Luby's algorithm as a basis~\cite{Luby1986}, but an agent is selected only if it wins a pairwise random trial against each of its neighbors. Each pairwise trial works by simultaneously exchanging random bits with the neighbor and taking the exclusive or of the shared and received private bits to obtain random bits that are known to both parties and cannot be controlled by either one alone.

To our knowledge, this represents the first distributed implementation of a mechanism that resolves the issue of relaying information without assumptions about network structure or all-to-all communication. We hope that our techniques will be useful in designing more sophisticated distributed mechanisms.

In Section~\ref{sec:additional-mechanisms} we present several examples of incentive-compatible mechanisms for classic graph problems. These mechanisms illustrate different algorithmic techniques, e.g.\ the local ratio algorithm~\cite{baryehuda81linear}.

We assume that the agents' private weights come from a bounded set of possible values, as the greedy strategy of choosing local maxima can take linear time to converge otherwise. In Section~\ref{sec:discretisation} we show that the mechanism designer can enforce such discrete bids, even if the true weights are continuous, without losing incentive-compatibility. The tradeoff is a small decrease in the approximation ratio.

\subsection{Open questions}

Our work initiates the study of local distributed mechanisms for graph problems. There are many interesting open problems both in the trusted and consistent settings.
\begin{enumerate}
	\item \emph{Mechanisms with small messages.} We show that there is general method for computing the critical prices defined by a monotone distributed optimisation algorithm. However, this general method relies on the unbounded bandwidth provided by the LOCAL model. For most applications this is not practical, and it would be important to understand e.g.\ if the critical prices of the classical greedy set cover algorithm, used in Section~\ref{sec:dominating-set}, can be computed efficiently in the CONGEST model.
	\item \emph{Monotonicity in distributed optimisation.} Myerson's lemma automatically gives a distributed mechanism based on any monotone distributed local optimisation algorithm. This means that it is important to understand which distributed algorithms are monotone and which are not. We propose that this is an important property that is of interest in any new proposed distributed optimisation algorithms. Similarly it is interesting to understand which algorithmic design techniques contradict this property. For example, algorithms working on some ``good'' subset of the input most likely are not monotone, but, as we have shown in Section~\ref{sec:discretisation}, grouping nodes by their weight can be effective and retain truthfulness.
	\item \emph{Mechanisms with large degrees.} We have studied distributed mechanisms with the assumption that the underlying graph has small maximum degree $\Delta$. This is because we use colouring for tie-breaking, forcing a linear in $\Delta$ term into the running time of each of our algorithms. This appears to be difficult to avoid, so we ask whether there exist distributed mechanisms with non-trivial optimisation guarantees and sublinear-in-$\Delta$ running times?
	\item \emph{Price of truthfulness.} Since distributed mechanisms operate in a narrower design space, it is natural to ask if non-monotone algorithms can achieve better approximations than monotone algorithms. It would be interesting to understand if there are natural cases where there is a real cost in locality associated with requiring truthful distributed algorithms.
\end{enumerate} 

\subsection{Related work}

There exist two lines of research that are related to the two aspects of our work. \emph{Distributed Algorithmic Mechanism Design (DAMD)} was introduced for studying mechanism design for distributed systems~\cite{feigenbaum01sharing,feigenbaum02bgp}. The goal is to have distributed implementations of mechanisms run by the strategic agents themselves. The problems studied in this line of research have been fundamentally \emph{global} in nature (i.e.\ computation requires gathering information from the whole network to one point). These include problems such as routing~\cite{feigenbaum02bgp,roughgarden2007routing}, multicast cost sharing~\cite{Acher-approximation,feigenbaum01sharing}, scheduling~\cite{carroll11distributed} or leader election~\cite{afek14building}. For a survey on DAMD, we refer to a book chapter by Feigenbaum, Schapira, and Schenker~\cite{feigenbaum07book}.

Second, in algorithmic mechanism design there has been some work on mechanisms for local graph problems. Studied problems include weighted vertex cover~\cite{calinescu15bounding,elkind07frugality} and set cover~\cite{devanur03strategy,li10mechanism}. These works have been, however, from the centralised perspective.

The \emph{Vickrey-Clarke-Groves (VCG)} mechanism is a generic framework for solving any optimisation problem truthfully~\cite{vickrey61auction,clarke71,groves73}. This mechanism does not have a fast distributed implementation, as it requires the computation of an optimal solution. This requires linear time in the LOCAL model for any reasonable optimisation problem.

Distributed graph algorithms for the tasks we study here have received a lot of attention, but most works do not consider any game-theoretic elements. Collet, Fraigniaud, and Penna~\cite{collet18equilibria} study the equilibria of simple randomised algorithms for basic symmetry breaking tasks. Hirvonen, Schmid, Schmid, and Chatterjee~\cite{hirvonen23convergence} present games where best-response dynamics simulate distributed algorithms.

The incentive-compatible deferred acceptance algorithm~\cite{gale62sm} for computing a stable matching has a natural implementation as a distributed proposal algorithm. While the operations are local, the algorithm can take $\Omega(n^2)$ rounds to converge~\cite{ostrovsky2015fast} and requires $\Omega(n)$ rounds in the LOCAL model~\cite{floreen10stable}. To circumvent this, Flor{\'{e}}en, Kaski, Polishchuk, and Suomela~\cite{floreen10stable} and Ostrovsky and Rosenbaum~\cite{ostrovsky2015fast} study distributed algorithms for computing so-called almost stable matchings. These algorithms are no longer incentive-compatible. Hirvonen and Ranjbaran showed that the variant where one side has common preferences admits a local distributed algorithm~\cite{hirvonen24fast}.

\section{Preliminaries} \label{sec:preliminaries}

\subsection{Game theory and mechanism design} \label{ssec:gt}

We study systems modelled as graphs $G = (V,E)$ where each node $v \in V$ corresponds to a strategic agent. Each agent $v$ has a private (hidden) type $\theta(v)$. We will focus on \emph{single-dimensional} agents where the type is a single parameter $w(v)$, the weight of $v$. Following the \emph{revelation principle}, we assume that mechanisms are direct revelation mechanisms: each agent $v$ reveals a \emph{bid} $b(v)$ to the mechanism. The mapping from the true type of the agent to its bid is its \emph{strategy}. We use $b$ to denote the vector $(b(v) : v \in V)$ and use the standard shorthand $b_{-v}$ for the bid vector omitting the value of agent $v$, and $(x, b_{-v})$ for the vector equal to $b$, except the bid of $v$ has been switched to $x$. 

A mechanism is a game defined as follows. It has an allocation rule $M$ that, based on the bids $b$ and the structure of the graph $G$, computes an outcome $o$ from a set of \emph{feasible outcomes} $O$ (e.g.\ an independent set) and a payment $p_v \in \mathbb{R}$ for each agent $v$ (e.g.\ a price of joining for agents selected to an independent set). 

Each agent $v$ has a valuation function $u_v\colon O \to \mathbb{R}$ from a set of feasible outcomes $O$ to their valuations. In the example of modelling independent sets, the valuation $u_v(o) = w(v)$ if the agent is chosen in $o$ and 0 otherwise: the weight tells how much the agent gains from being selected. The \emph{utility} of $v$ is $U_v(o,p) = u_v(o) - p_v(o,b)$. In this work the goal of the mechanism will be to maximise the sum of the valuations of all agents (\emph{total welfare}), i.e.\ $\sum_{v \in V} u_v(o)$.

There are many possible solution concepts in mechanism design, but we will consider only pure \emph{Nash equilibria}~\cite{nash50equilibrium}. A mechanism is \emph{truthful} or \emph{incentive-compatible} if reporting the truth is always a (weakly) dominant strategy for the agents. That is, irrespective of $b_{-v}$, the agent $v$ maximises its utility by bidding $b(v) = w(v)$. This implies that truthful bidding is a Nash equilibrium in the game formed by the mechanism. When we say that a mechanism computes a solution, we mean that we assume that the agents bid truthfully (when truth is a weakly dominant strategy) and then the mechanism computes the solution corresponding to these bids.

\subsection{Local optimisation problems}

We study mechanism design in the context of distributed graph algorithms and will design mechanisms for graph problems on weighted graphs called \emph{local optimisation problems}. A solution to a local optimisation problem $f\colon V \to \Sigma$ assigns a label to each node from some alphabet $\Sigma$. The feasibility and quality of the solution are defined based on  local criteria as follows. First, these problems are \emph{locally checkable}~\cite{Naor1995}: there is a problem-specific constant $r$ such that each node can verify that the solution is correct in its $r$-neighbourhood, and the global solution is correct if and only if the solution is locally correct everywhere. That is, if the solution is not globally correct, there is at least one node that detects this in its local neighbourhood. Second, the quality of the solution can be represented as a sum of values over the nodes such that the value of each node depends only on its $r'$-neighbourhood, for some problem-specific constant $r'$. As an example, consider the \emph{weighted independent set} problem. To decide whether a marked set $I \subseteq V$ is independent, it is sufficient to see at distance $r=1$ from each node to verify that no two adjacent nodes are in $I$. To count the quality of the solution, it suffices to look at the labels of the nodes themselves (distance $r' = 0$) and sum the weights of the chosen nodes. When we return to the game-theoretic setting, this local valuation is the agent's valuation for the outcome represented by the solution $f$.

We will focus on one subset of local optimisation problems: \emph{binary maximisation and minimisation problems}. Here the task is to select a subset of nodes with respect to a locally checkable constraint maximising (or minimising) the total weight of the set.

For maximisation and minimisation problems, respectively, we say that a solution $f$ is an $\alpha$-approximation if 
\[
\sum_{v \in V} u_v(o^*) / \sum_{v \in V} u_v(o) \geq \alpha \text{ and } \sum_{v \in V} u_v(o) / \sum_{v \in V} u_v(o^*)  \leq \alpha
\] 
given an optimal outcome $o^*$ that maximises (minimises) the valuations of the agents.

We will assume that the values of the weight function $w$ come from the set $\{0,1,\dots,W\}$ for some known value $W$. In Section~\ref{sec:discretisation} we show how this assumption can be lifted (for a small tradeoff) by discretising a real-valued weight function.

\subsection{Distributed computing}

We study mechanism design in the standard distributed message passing models LOCAL and CONGEST~\cite{Peleg2000,Linial1992}. The system is modelled as a graph $G = (V,E)$ where each node $v \in V$ is a computational entity and each edge $\{u,v\} \in E$ is a communication link, allowing the exchange of messages between $u$ and $v$. Define $n = |V|$.

Each node is given a globally unique identifier in $[O(\operatorname{poly}(n))]$. Initially they only know their own identifier. The computation proceeds in synchronous rounds. In each round, nodes can send messages to their neighbours, receive messages from their neighbours, and update their state. There are no limitations on local computation. In the CONGEST model each message is limited to $O(\log n)$ bits, but in the LOCAL model there are no limits on the message size. In the CONGEST model we make the standard assumption that the input weights can be encoded with $O(\log n)$ bits.

At some point each node must stop and announce its own output. An algorithm has stopped once all nodes have stopped. The measure of complexity is the number of communication rounds. The running time of an algorithm is the worst-case round complexity over all graphs and identifier assignments. In this work we will study the complexity landscape when the maximum degree $\Delta$ of the input graph is small relative to the size of the network. 


\subsubsection{Implementing mechanisms in the distributed setting} \label{ssec:implementing-mechanisms}

We study the implementation of mechanisms as efficient distributed algorithms. We consider a general method of turning a (distributed) algorithm $A$ for a binary maximisation or minimisation problem (or more generally for a local optimisation problem) into a direct revelation mechanism: Given a weighted graph $G = (V,E,w)$, the weights $w$ are taken as the private types of the agents $V$. After the agents reveal their bids, the algorithm $A$ is used as the allocation rule $M$ for the mechanism, with the bids replacing private unknown weights. Finally, the the payments for all the agents (which are not part of the original algorithm $A$) are computed. How these payments are defined and computed is discussed in the next Section. We require that as output, each agent knows the part of the outcome that is needed to compute its own valuation and knows its payment.

For the computation of the allocation rule and the payments we study two variants of the model.
\begin{enumerate}
	\item In the \emph{trusted strategic LOCAL model} we assume the algorithm is run in a trusted environment: only the revelation of the inputs is strategic, and after this the system follows the algorithm faithfully. This model is used unless otherwise specified.
	\item In the \emph{consistent strategic LOCAL model} we also assume that the communication and computation is controlled by the strategic agents. We discuss this model in Section~\ref{sec:unreliable}, showing that some non-trivial mechanisms are implementable in this setting as well, but many are not.
\end{enumerate}
Corresponding variants of the CONGEST model are defined by adding the message size constraints to above models.


\subsection{Notation}

For a graph $G$ and a subset $S \subseteq V$ of nodes, let $G[S]$ denote the subgraph of $G$ induced by $S$. Two nodes $u$ and $v$ are \emph{neighbours} if $\{u,v\} \in E$. Let $N_G(v)$ denote the set of neighbours of $v$ in $G$ and $N^+_G(v) = N_G(v) \cup \{ v \}$. For a set $S \subseteq V$, define $N^+_G(S) = \cup_{s \in S} N^+_G(s)$.

\section{Local mechanisms and monotone optimisation} \label{sec:monotone-mechanism}

In this section we discuss the known connection between the existence of truthful auctions and monotone optimisation algorithms, and how this has a simple and natural extension to distributed computing. We also discuss how mechanisms arising from this connection are automatically able to compute the payments with the same asymptotic complexity in the LOCAL model.

Myerson's lemma states that in single-parameter environments (i.e.\ each agent has a type that is characterised by a single real number) only monotone algorithms yield truthful mechanisms and all such algorithms have a payment rule that makes them incentive-compatible as mechanisms. The mechanisms arising from binary maximisation and minimisation problems, as defined in the previous section, are covered by Myerson's Lemma.


An allocation rule $A$ is \emph{implementable} if there exists a rule for computing the payments such that the resulting mechanism is truthful. For maximisation problems, an allocation rule is \emph{monotone} if for each agent $v$, when other agents' bids remain constant, the outcome $o(v, (x,b_{-v}))$ is non-decreasing in $x$. For binary maximisation problems an algorithm $A$ is monotone if whenever a node $v$ is selected with bid $b(v)$, then it is also selected with any bid $b'(v) \geq b(v)$ (when the other bids remain constant). Similarly an algorithm for a binary minimisation problem is monotone if whenever a node $v$ is selected with bid $b(v)$, then it is also selected with any bid $b'(v) \leq b(v)$.

\begin{theorem}[Myerson's Lemma~\cite{myerson81optimal}] \label{thm:myerson}
	Consider a single-parameter environment.
	\begin{enumerate}
		\item An allocation rule $A$ is implementable if and only if it is monotone.
		\item There exists an explicit formula for the payment rule.
	\end{enumerate}
\end{theorem}

For the special case of binary maximisation and minimisation problems, the payments are given by the \emph{critical prices}: $b^*(v)$ such that $v$ is selected with any bid $b(v) \geq b^*(v)$ and is not selected with any bid $b(v) < b^*(v)$. The existence of a critical price follows directly from the monotonicity of $A$. The critical price is used to define the payment function: in binary maximisation problems the payment is defined as $p(v) = b^*(v)$ if $v$ is selected and 0 otherwise, and in binary minimisation problems it is defined as $p(v) = -b^*(v)$ if $v$ is selected and 0 otherwise.

We show that any monotone distributed algorithm can be turned into a mechanism.
\begin{theorem} \label{thm:main}
	Let $A$ be a monotone distributed algorithm for a binary maximisation or minimisation problem $P$ in the LOCAL model. Then the mechanism that uses $A$ as the allocation rule and critical prices as payments is incentive compatible and can be computed with the same time complexity as $A$ in the trusted LOCAL model.
\end{theorem}

While the incentive compatibility follows from Myerson's Lemma, we give an independent proof for our special case. In the next section we show that the critical prices can also be computed with the same time complexity, completing the proof of Theorem~\ref{thm:main}.

\begin{lemma} \label{lem:monotone}
The mechanism from Theorem~\ref{thm:main} is incentive compatible.
\end{lemma}

\begin{proof}
Consider maximisation problems. Consider an arbitrary graph $G = (V,E)$ and some $v \in V$. First assume that truthful bidding $b(v) = w(v)$ results in $v \in S$. Since $A$ is monotone, we have that the critical price $b^*(v) \leq b(v)$ and $v \in S$ for all $b^*(v) \leq b'(v)$. Therefore all bids $b'(v) \geq b^*(v)$ result in the same utility $U_v(S,p) = w(v) - b^*(v)$. If $v$ bids $b'(v) < b^*(v)$ it is no longer selected by monotonicity. Its utility is $0$, which is at most the utility of the truthful bid. Conversely, if $w(v) < b^*(v)$, the utility of truth and any other bid $b'(v) < b^*(v)$ is 0. Bidding $b'(v) \geq b^*(v)$ results in total utility $w(v) - b^*(v) \leq 0$.
\end{proof}
The proof for minimisation problems is analogous.

\subsection{Critical price computation in the trusted LOCAL model}

A mechanism consists of two important parts: computing the assignment and computing the payments. Given a monotone distributed algorithm with running time $T(n,\Delta,W)$ for computing the outcome, the critical prices can be computed with the same time complexity in the LOCAL model.

\begin{lemma} \label{lem:local-payments}
	Given a $T(n,\Delta,W)$-time monotone distributed algorithm $A$ for a binary optimisation problem the critical prices can be computed with the same distributed time complexity.
\end{lemma}

\begin{proof}
	By definition the output of a node depends only on its $T(n,\Delta,W)$-neighborhood. Each node can therefore gather the full $T(n,\Delta,W)$-neighborhood and simulate $A$ with all possible bids $b(v)$ to find the critical price $b^*(v)$.
\end{proof}

In the CONGEST model it is no longer feasible to collect the full neighbourhood to each node. This means that computing the payments becomes a non-trivial challenge. In the next section we present a mechanism for computing independent sets and show that in this case the critical prices can also be computed in the CONGEST model.

\section{A truthful mechanism for weighted independent set} \label{sec:independent-set}

We first present a mechanism for a fundamental task in distributed computing: finding an independent set of large weight. Weighted independent set models a scenario where each agent wants to take an action, and the graph represents the set of conflicts: neighbours cannot act in parallel.

The mechanism computes an independent set $I$. The weights $w(v)$ represent valuation for independent activation (agent is in $I$ and no neighbour in $I$). If an agent and its neighbour are both in $I$, or an agent is not in $I$, its valuation is 0. Recall that we assume that the weights are integers in $\{0,1,2,\dots,W\}$ for some known parameter $W$.

\begin{theorem} \label{thm:is-mechanism}
	There is an incentive-compatible distributed mechanism for the maximum weight independent set problem that computes a $\Delta$-approximation and runs in $O(\Delta W + \log^* n)$ rounds in the trusted CONGEST model.
\end{theorem}

The mechanism uses a simple greedy algorithm~\cite{sakai03note}. Compute a $(\Delta+1)$-colouring $\phi$ for tie-breaking, reveal the private inputs, and then repeat the following until all nodes have stopped: select the set of local maxima according to $<_{\phi}$, add them to the independent set, and remove them and their neighbours from the set of active nodes. We need to show that the algorithm has the claimed approximation ratio (Lemma~\ref{lem:is-approx}), the mechanism is monotone (Lemma~\ref{lem:greedy-is-monotone}), and that the payments can be computed fast (Lemma~\ref{lem:independent-set-property}).

\begin{lemma} \label{lem:is-approx}
	The greedy algorithm that repeatedly picks a maximal independent set of local maxima computes a $\Delta$-approximation of maximum weight independent set.
\end{lemma}

The idea for the proof can be found for example in Sakai, Togasaki, and Yamazaki~\cite{sakai03note}, Lemma~3.6. We must show that the distributed implementation works even when we pick many local maxima simultaneously. The proof is given in Appendix~\ref{app:is-apx-proof}.

\begin{lemma} \label{lem:greedy-is-monotone}
	The greedy weighted independent set algorithm is monotone.
\end{lemma}

\begin{proof}
	Assume $v$ bidding $b(v)$ is selected in the $i$th step of picking local maxima. Consider an alternative execution where $v$ bids $b'(v) > b(v)$ and other bids remain constant. In each step $j = 1, \dots, i-1$ one of two things can happen: either no agent in $N^+(v)$ is selected or $v$ is selected. Since we know that none of the neighbours are picked with bid $b(v)$, they are never picked with the larger bid, as they cannot become local maxima. This implies that $v$ is either picked earlier or picked in round $i$ when it is a local maximum.
\end{proof}

\subsection{Computing critical prices with small messages}

We claim that to compute the critical price $b^*(v)$ the following suffices: each agent $v \in I$ arranges its neighbours with smaller bids in the descending order $u_1, \dots, u_k$ according to $<_{\phi}$. Each neighbour $u_i$ sends the following information: does it have a neighbour $x \neq v$ such that $b(x) >_{\phi} b(u_i)$ and $x \in I$? Agent $v$ computes the first neighbour $u_i = u^*(v)$ such that $u^*(v)$ \emph{does not} have such a neighbour: $b^*(v) = b(u^*(v))$.

\begin{lemma} \label{lem:independent-set-property}
	In the greedy weighted indepedendent set mechanism, the critical price of agent $v$ is determined by the highest bidding neighbor $u$ such that if $b(u) <_{\phi} b(v)$ we have $v \in I$ and $u$ has no other neighbor $w$ with $b(w) >_{\phi} b(u)$ and $w \in I$.
\end{lemma}

\begin{proof}
	The greedy algorithm repeatedly picks the local maxima of $<_{\phi}$ and therefore the output of an agent depends exactly on the output of its neighbors $u$ with $b(u) >_{\phi} b(v)$. Let $G_{b,\phi} = (V,E_{b,\phi})$ denote a directed graph constructed as follows: for each edge $\{u,v\} \in E$, the dependency graph has a directed edge $(u,v) \in E_{b,\phi}$ if and only if $b(u) <_{\phi} b(v)$. We say that $v$ is a \emph{successor} of $u$ if there is a directed path from $u$ to $v$ in $G$, and $u$ is a \emph{predecessor} of $v$. A node $u \in N(v)$ is a closest (direct) predecessor of $v$ if it is a predecessor of $v$ and there is no predecessor $u'$ of $v$ in $N(v)$ with $b(v) >_{\phi} b(u') >_{\phi} b(u)$.  
	
	First, assume that given some bid $b(v)$ we have $v \in I$ and let $u \in N(v)$ be a closest direct predecessor of $v$. Now if $u$ has a neighbor $w \in I$ with $b(w) >_{\phi} b(u)$, then $w$ is not a predecessor of $v$. This implies that if $v$ bids $b'(v) <_{\phi} b(u)$ this does not affect the output of $w$ and we have $u \notin I$. Conversely, assume again that $u$ is a closest direct predecessor of $v$, but there is no $w \in I \cap N(u)$ with $b(w) >_{\phi} b(u)$. No successor of $u$ is a predecessor of $v$, so if $v$ bids $b'(v) <_{\phi} b(u)$ instead, then all successors of $u$ output $\notin I$ and we have $u \in I$.    
\end{proof}

\subsection{Proof of Theorem~\ref{thm:is-mechanism}}

\begin{proof}
	By Lemma~\ref{lem:is-approx} the mechanism correctly computes a $\Delta$-approximation. By Lemma~\ref{lem:greedy-is-monotone} the mechanism is monotone, so from Myerson's Lemma it follows that it is truthful.
	
	The algorithm consists of computing a $(\Delta+1)$-coloring, repeatedly picking local maxima until all agents have stopped, and computing the payments. A $(\Delta+1)$-coloring can be computed in $O(\Delta + \log^* n)$ rounds~\cite{maus20linial}. There are $(\Delta+1)W$ weight classes in total (initially $W$ classes in the input and each is subdivided into at most $(\Delta+1)$ classes by the tie-breaking) and in each step all agents who are local maxima are eliminated, so after $O(\Delta W)$ rounds all agents have been eliminated. After the algorithm has stopped, each agent $v \notin I$ can inform its neighbors in $I$ whether they have another larger neighbor in $I$. This allows each $u \in I$ to compute its critical price by Lemma~\ref{lem:independent-set-property}.
\end{proof}

\subsection{Optimality}

The worst-case guarantee of the greedy algorithm is factor $O(\log \Delta)$ from optimum among distributed algorithms with running time~$o(\log n)$, even in the unweighted case. This is because there exist $d$-regular graphs of logarithmic girth with no independent sets of size $\omega(n \log d / d)$~(see e.g.\ \cite{alon10constant}, Theorem 1.2). These graphs are indistinguishable from bipartite graphs of logarithmic girth, and therefore an $o(\log n)$-time distributed algorithm cannot compute a solution of size $\omega(n \log d / d)$ in the latter.


\section{Mechanisms with unreliable communication} \label{sec:unreliable}

In this section we study what can be done if we remove the assumption of reliable communication and algorithm execution. That is, instead of the agents only controlling the information that they reveal, they also run the algorithm itself, including choosing what messages to send.

The fundamental difficulty in such a model is that agents can no longer trust any information about the network that is relayed by its neighbours. This makes algorithm design very difficult. In Section~\ref{ssec:im-dep} we formalise this observation and characterise a class of algorithms that are not truthful in the strategic communication and computation setting.

Despite this impossibility result it is possible to design non-trivial algorithms. In Section~\ref{ssec:alg-sc-mwis} we present a strategy-proof mechanism for weighted independent set approximation. This mechanism uses a different primitive for tie-breaking compared to the mechanism from Section~\ref{sec:independent-set}: it is slower, uses more colours, and requires the use of randomness.

Distributed algorithms in the strategic setting require a weaker solution concept than dominant strategies: if all other agents deviate, it is typically the case that a single agent should also deviate. The natural concept is \emph{ex-post Nash equilibrium}: that is, Nash equilibria assuming we know the strategies of the other agents, but not their private information (valuations).

\subsection{Modelling strategic communication and computation}

In this section we study the setting where the strategic agents are responsible for running the mechanism. Instead of following a fixed algorithm, in each step each agent can perform arbitrary computational steps and send arbitrary messages. There is no cost to these steps and the utility is computed based on the outcome and the payment. We call any step where the agent does not follow the intended algorithm a \emph{deviation}.

\textbf{Consistency.} In the context of mechanism design, the idea is that there is an entity, representing the world outside the mechanism, that collects or makes payments to the agents. We assume that these payments are made in a way that is verifiable and model it by assuming that after the execution, the payment of an agent is consistently visible to all of its neighbours. Similarly, we assume the outcome is an action that is also consistently visible to all neighbours of each agent. This means that an agent cannot lie inconsistently about its computed outcome: all lies must result in the same outcome and payment.

\textbf{Detecting deviations.} While agents can lie, we assume that if their lies are detected, other agents can always report this to the mechanism. By a detectable deviation we mean that there is provable deviation from the algorithm by a specific agent based on the information received by a single agent. If such a deviation is detected, we assume all cooperating agents are guaranteed a payment that incentivises them to report, and all agents that deviate and are detected receive a penalty that always disincentivises them. We further assume that agents are unable to forge evidence of a detectable deviations, e.g.\ through the use of signed messages. Thus we will assume that agents only lie in ways that are undetectable. As an example, agents will not change their announced outputs, as this would be easy to detect. However, agents may e.g.\ lie about their private inputs, and even give different lies to different agents, as long as these agents cannot compare the relevant messages.

\textbf{Prefer cooperation.} We assume that if all things are otherwise equal, agents will always prefer cooperation. That is, if following the algorithm gives the same utility as some deviation, then the agent will choose to follow the algorithm.

\textbf{Ex-post Nash equilibrium.} Following the algorithm is a (weak) ex-post Nash equilibrium if each agent maximises its utility by revealing its true valuations and following the algorithm assuming the other agents do so. This is a weaker notion than the dominant-strategy incentive compatibility used in the previous sections, but even under this notion the agents do not need to know anything about the private information of the other agents, only that they are following the algorithm without deviations.  We refer to the book by Nisan et al.\ for a more formal exposition on ex-post Nash equilibria~\cite[Chapter 9]{nisan07algorithmic}.

\subsection{A mechanism for weighted independent set with strategic communication} \label{ssec:alg-sc-mwis}

In this section we show that the independent set algorithm from Section~\ref{sec:independent-set} can be adapted to work in the consistent LOCAL model. Recall that the algorithm consists of two phases: in the first phase, a tie-breaking coloring $\phi$ is computed. Then in the second phase agents make their bids and an independent set is computed using the simple greedy local maxima selection.

The main challenge in adapting the algorithm is the tie-breaking. If we use an arbitrary coloring algorithm, the agents could manipulate it in complicated ways. This could easily lead to situations where an agent could gain an advantage. For example, suppose agent $v$ has $w(v) < w(u)$ for some neighbour $u$, and $u$ in turn has neighbour $u'$ with $w(u') = w(u)$. Further, suppose $\phi$ decides which of $u$ and $u'$ joins $I$. Now if $v$ can manipulate $\phi$ so that $u'$ joins it can also potentially join $I$ and improve its utility.

To avoid this, we need to design a subroutine where the choices of single agents do not affect the colours chosen by other agents. The key idea is to do pairwise random trials between neighbours where each agent generates a random number by choosing a bit string $X$ u.a.r., receiving another string $Y$ from its neighbour, and taking the bitwise exclusive OR $Z = X \oplus Y$ as its input. To prevent an agent manipulating its bit, each agent must send both its private string $X$ and its string intended for its neighbour, $Y$, over to the neighbour in the initial message. This allows both nodes to know both resulting strings. We say that a node $v$ wins its trial with $u$ if $Z_v > Z_u$. We say that a node wins in a round if it wins all of its neighbours.

Let $r$ be a positive integer parameter. While there is an active neighbour, for steps $i = 0,1,2,\dots$ each node $v$ repeats the following.
\begin{enumerate}
	\item Round 1: For each neighbour $u$, choose two $r$-bit strings $X_{v,u}$ and $Y_{v,u}$ uniformly at random. Send $X_{v,u}$ and $Y_{v,u}$ to $u$. For each neighbour $u$, receive $X_{u,v}$ and $Y_{u,v}$ and compute $Z_{v,u} = X_{v,u} \oplus Y_{u,v}$ and $Z_{u,v} = X_{u,v} \oplus Y_{v,u}$.
	\item Round 2: If $v$ is active and for all $u \in N(v): Z_{v,u} > Z_{u,v}$, set $\phi(v) = i$. Set status as inactive and send a message to each neighbour. If all neighbours are inactive stop and set $\phi(v) = i$.
\end{enumerate}

\begin{lemma} \label{lem:pairwise-invariant}
	For any fixed agent $v$, its non-detectable deviating actions can only affect its own outcome. 
\end{lemma}

\begin{proof}
	Here we use the assumption that the other agents are following the algorithm. Since each neighbour $u$ chooses its bits randomly, by taking $\oplus$ we ensure that both $Z_{u,v}$ and $Z_{v,u}$ are uniform random strings no matter what strings $v$ chooses. In each step, each neighbour $u$ wins the trial with probability $(1-\varepsilon)/2$ for any choice of $X_{v,u}$ and $Y_{v,u}$, where $\varepsilon = \varepsilon(r)$ is the probability of the bit strings being equal. This implies that in each step, each agent $u$ has probability $((1-\varepsilon)/2)^{\deg(v)}$ of winning, again independent of the choices of the bits of $v$. Finally, since the colour chosen by $u$ only depends on the step $i$ at which it wins for the first time, it is also not dependent on the choices of $v$.
	
	If $v$ deviates by not choosing its bits randomly, this does not affect the probabilities of other agents. If $v$ tries to announce itself as a winner in a step that it did not win, at least one other agent will detect the deviation. Since the communication in the algorithm consists of sending bit strings and possibly announcing a win, all other deviations are detected. 
\end{proof}

\begin{lemma} \label{lem:pairwise-correct}
	For a constant maximum degree $\Delta$, the pairwise trial algorithm computes a proper colouring with $O(\log n)$ colours in $O(\log n)$ rounds with high probability.
\end{lemma}

\begin{proof}
	In each step $i=0,1,2,\dots$ the winning agents form an independent set, and therefore the colors are always proper. For a constant $\Delta$, each agent has a constant probability of winning in each step (step consists of two rounds), and thus wins at least once in $O(\log n)$ rounds w.h.p. By a union bound all agents win in $O(\log n)$ rounds w.h.p.
\end{proof}

\begin{theorem} \label{thm:consistent-is}
	There exists a mechanism that computes a $\Delta$-approximation of maximum-weight independent set in $O(W \log n)$ rounds when $\Delta$ is constant. Following the mechanism is a weak ex-post Nash equilibrium in the consistent strategic LOCAL model.
\end{theorem}

\begin{proof}
	By assumption, detectable deviations are never preferable. Therefore it remains to show that no unilateral undetectable deviation can result in a gain in utility, assuming other agents follow the algorithm without deviation.
	
	From Lemma~\ref{lem:pairwise-invariant} we know that the agents cannot manipulate the tie-breaker $\phi$ beyond (some) manipulations to their own color. We will show that this and deviations in the greedy selection algorithm cannot improve utility. In the greedy algorithm each agent initially announces its bid (can deviate by announcing different values to different neighbours). Then in each round the local maxima announce themselves and stop. Then the neighbours of those local maxima announce that they are not in the set and stop. Any communication outside these messages is a deviation.
	
	We will show that deviations do not help using a case analysis based on the outcome of non-deviating behavior. Recall from Lemma~\ref{lem:independent-set-property} that the output of agent $v$ in the greedy algorithm only depends its successors, i.e. other agents connected to it by directed paths of increasing bids.
	\begin{enumerate}
		\item Case $w(v) <_{\phi} b^*(v)$: Let $u$ be one of the critical neighbors of $v$. The only way to join $I$ and improve utility is to lie such that $u$ is no longer critical: that is, agent $v$ must lie so that some successor $u'$ of $u$ changes its output from $u' \notin I$ to $u' \in I$. To affect the output of an agent, $v$ has to affect that agent or one of its successors. But since $u$ is a critical neighbor, for all successors $w$ of $u' \in N(v)$ it holds that $w \notin I$. Therefore increasing $b(v)$ and $\phi(v)$ does not change the output of these agents, and no lie that causes $u$ to become non-critical exists.
		
		\item Case $w(v) >_{\phi} b^*(v)$: Let $u$ denote the critical neighbour of $v$; if there is no critical neighbour, $v$ cannot improve its utility in any way. By definition, if $v$ changes $\phi(v)$ or $b(v)$ so that $b(v) <_{\phi} b(u)$ then $v \notin I$ and its utility is not improved. The other option is to lie to some other agent to cause a change in the output so that $u$ becomes non-critical. Since $v \in I$ no lie upwards will change the outputs of its neighbours $u' \in N(v)$, as they all have $u' \notin I$. If $v$ lies to some neighbour $u'$ so that it changes to $u' \in I$ then by consistency we must have $v \notin I$, i.e. no gain in utility.
	\end{enumerate}	
	
	Once agent and its neighbors have stopped, each agent knows the bids and outputs of its neighbors. Each agent $u \notin I$ sends a message to each of its neighbors $v \in I$, indicating if it has another neighbor $w \in I$ with $b(w) >_{\phi} b(u)$. Agents have no incentive to lie, as this only affects the payments of other agents. Each agent $v \in I$ can then decide its critical price and payment, as shown in Lemma~\ref{lem:independent-set-property}. If an agent $v \notin I$ that was not chosen tries to output $v \in I$ by the maximality of the independent set there is a neighbor that detects a deviation. If it tries to announce a payment $b(v) <_{\phi} b^*(v)$, then, by consistency, it has to announce the same value to all neighbors, and by definition there is a neighbor $u$ that detects it could have joined $I$ if $v$ had bid $b(v)$.
\end{proof}

\subsection{Impossibility of truthfulness with dependent information} \label{ssec:im-dep}

We will finish this section by showing that the algorithm from the previous section has the structure that it does for a reason: any algorithm that is flexible enough and uses information relayed by other agents cannot be an ex-post Nash equilibrium.

We observe that if an agent $v$ can feed different lies to its neighbors such that all neighbors believe the utility of $v$ should be higher than in the real system, then it has incentive to present these lies. 

Consider input-labelled graphs with unique identifiers and $b$ denoting bids of the agents. Let $\mathcal{G}$ denote some family input-labelled graphs.

For a graph $G$ and an edge $\{u,v \}$ we say that the \emph{$(u,v)$-restriction of $G$ at $v$} is the connected component of $u$ in $G \setminus \{v\}$. This naturally induces restricted input functions such as $b$ and the identifiers. These represent the information that $u$ could potentially receive without being relayed by $v$. Conversely, given $G'$, a $(u,v)$-restriction of a $G \in \mathcal{G}$ at $u$, a \emph{$(u,v)$-extension of $G'$ at $u$} is any $G'' \in \mathcal{G}$ such that $G''$ can be constructed by taking some input-labelled graph $B$ containing $v$ and adding the edge $\{ u,v \}$ s.t. all input labels agree on $G' \cup \{v \}$ (i.e. the difference is in the part hidden "behind" $v$).

We say that an algorithm $A$ in a family of graphs $\mathcal{G}$ is \emph{cross-dependent}, if the following holds. There exists a graph $G \in \mathcal{G}$ and a node $v \in G$ such that
\begin{enumerate}
	\item for every $u,u' \in N(v)$ the $(u,v)$-restriction and the $(u',v)$-restriction of $G$ at $v$ have an empty intersection, and
	\item for all $u \in N(v)$ there exists an $(u,v)$-extension of the $(u,v)$-restriction of $G$ at $v$ such that all extensions agree on the outcome of $v$ under $A$, all extensions agree on the payment of $v$ under $A$, and the utility of $v$ is strictly higher than in $G$.
\end{enumerate} 

\begin{theorem}
	If a mechanism $A$ is cross-dependent on an input-labelled family of graphs $\mathcal{G}$ then it is not an ex-post Nash equilibrium in the strategic LOCAL model.
\end{theorem}

\begin{proof}
	From the definition it follows that there exists a graph $G$ where an agent $v$ can feed its neighbours undetectable lies about its neighbourhood, given by the extensions, such that the algorithm $A$ will produce an outcome that is preferable to $v$ over the outcome the algorithm would have produced on $G$ without deviations.
\end{proof}

Note that the general form of the result relies on the agent potentially predicting what kind of information its neighbors will reveal, i.e. a worst-case analysis. Below we give an example of a simple algorithm where the lie is much easier to construct.

\subsubsection{Examples}

This characterisation captures a large class of algorithms. To illustrate it, we will give two examples below. It should be noted that the reason that our mechnism for weighted independent set avoids this impossibility is because all computation in it is based on pairwise comparisons instead of more general properties of agents' neighborhoods.

\paragraph{Example 1: Degree-dependent score function for greedy weighted independent set.} In the previous section we presented a variant of the greedy mechanism for weighted independent set that is an ex-post Nash equilibrium. In principle we could use a different function, such as the weight of the agent relative to its degree $w(v)/\deg(v)$~\cite{sakai03note} (we will discuss a use case for such a scoring function in Section~\ref{sec:dominating-set}). However, this mechanism is vulnerable to a simple lie: assuming that the different incident restrictions of an agent $v$ do not overlap, an agent can tell its neighbors that it has a very low degree, boosting its score $w(v)/\deg(v)$. This allows it to have a lower critical price.

\paragraph{Example: The VCG-mechanism.} The Vickrey-Clarke-Groves mechanism is one of the seminal results in mechanism design. It states that by computing the optimum solution and setting the payment of each agent $v$ as its \emph{externality}, i.e.\ the amount the utility of the other agents could be improved if we did not count $v$. It can be applied to essentially any optimisation problem, but is in practice computationally infeasible due to the requirement of computing the optimal solution to hard optimisation problems. If we consider the weighted independent set as an example, a focal agent $v$ can lie to each neighbor $u$ that in its $(u,v)$-extension choosing $v$ is "cheap" (there does not exist a heavy independent set that does not include $v$). This implies that the payment of $v$ is small.


\section{Mechanisms in the trusted LOCAL model} \label{sec:additional-mechanisms}

In this section we present three additional mechanisms for optimisation problems in the trusted LOCAL model. These mechanisms are examples of different algorithmic techniques that are still monotone as required by the Myerson's Lemma. Here $n$ denotes the number of nodes, $\Delta$ the maximum degree, and $W$ the number of possible private values.
\begin{enumerate}
	\item Minimum-weight vertex cover (MWVC): we give a mechanism for 2-approximation that runs in $O(\Delta + \log^* n)$ rounds in LOCAL model (Theorem~\ref{thm:vertex-cover-mechanism}). It is based on the local ratio algorithm~\cite{baryehuda81linear} which has no dependency on $W$ in the running time.
	\item Minimum-weight dominating set (MWDS): we present a mechanism that computes a $(1 + \ln (\Delta+1))$-approximation and runs in $O(\Delta^3 W + \log^* n)$ rounds in the LOCAL model (Theorem~\ref{thm:dominating-set-mechanism}). It is based on the greedy set cover algorithm~\cite{johnson74approximation,chvatal79greedy}.
	\item Slot assignment (a type of weighted graph colouring): we present a mechanism that computes an $O(\Delta)$-approximation and runs in $O(\Delta W + \log^* n)$ rounds (Theorem~\ref{thm:slot-auction-mechanism}). It uses a simple greedy algorithm. This mechanism showcases the application of Myerson's Lemma to problems more general than selecting a subset of nodes.
\end{enumerate}

\subsection{Mechanism for minimum-weight vertex cover} \label{sec:vertex-cover-mechanism}

In this section we study the weighted vertex cover problem. Given a weighted graph $G = (V,E,w)$ the task is to find a subset $C \subseteq V$ of nodes with small weight such that every edge has at least one endpoint in $C$. We show the following theorem.

\begin{theorem} \label{thm:vertex-cover-mechanism}
	There is an truthful distributed mechanism that computes a 2-approximation of minimum-weight vertex cover in $O(\Delta + \log^* n)$ rounds in the LOCAL model.
\end{theorem}

The mechanism uses the monotone local ratio algorithm for vertex cover due to Bar-Yehuda and Even~\cite{baryehuda81linear}. We require that each edge is covered. The local ratio algorithm also works for the \emph{prize-collecting vertex cover} problem where edges have a weight that must be paid if they are left uncovered~\cite{baryehuda04local}. The valuation is $-w(v)$ for each selected agent $v$ and 0 for each non-selected agent.

The following algorithm computes a 2-approximation of minimum weight vertex cover:
\begin{enumerate}[noitemsep]
	\item Initialise as follows: for each edge $e$ set $c(e) = 0$. For each node set $t(v) \leftarrow w(v)$. Let $C = \emptyset$ denote the vertex cover.
	\item Colour the line graph properly with $2\Delta-1$ colours.
	\item Proceed by colour classes of the line graph. At step $i$, for each uncovered edge $e = \{u,v\}$ with colour $i$, let $m(e) = \min \{ t(u), t(v) \}$. Set $t(u) = t(u) - m(e)$, $t(v) = t(v) - m(e)$, and $c(e) = m(e)$. Add all $v$ with $t(v) = 0$ to $C$.
\end{enumerate}

\begin{lemma}
	The local ratio algorithm for vertex cover is monotone.
\end{lemma}

\begin{proof}
	First, observe that the the order in which the edges are considered does not depend on the weights of the agents.
	
	Now assume that node $v$ is chosen with weight $w(v)$. This implies that at some step $i$ its remaining weight $t(v)$ was reduced to 0. If the weight was $w'(v) < w(v)$ then if $t(v) > 0$ by step $i$, it is also reduced to 0 at step $i$.
\end{proof}

\begin{lemma}
	The algorithm computes a 2-approximation of the minimum-weight vertex cover.
\end{lemma}

\begin{proof}
	The proof is essentially the same as for the original local ratio algorithm ~\cite{baryehuda81linear}, as each colour class of edges can be seen as being processed in an arbitrary order: each ordering will produce the same end result (endpoints of an edge cannot cover another edge with the same colour). We provide the full proof for completeness.
	
	First, the output forms a vertex cover, as at least one of the endpoints of each edge joins the cover. We have the following invariant: after an edge $e = \{u,v\}$ is processed, we charge each of the endpoints $m(e)$. This implicitly creates a new instance with weights $t$ such that edge $e$ can be covered with cost 0: the optimum cost goes down by $m(e)$. Since the cost of the solution is at most the charges made at the nodes, and this is $2m(e)$ at each step, we have that the total cost is at most 2 times the optimum.
\end{proof}

\begin{lemma}
	The mechanism can be implemented with running time $O(\Delta + \log^* n)$ in the LOCAL model.
\end{lemma}

\begin{proof}
	The line graph can be coloured with constant multiplicative overhead with the best $(\Delta+1)$-colouring algorithm in $O(\sqrt{\Delta \log \Delta} + \log^* n)$ rounds~\cite{maus20linial}.
	In the loop each colour class can be handled in one round: the endpoints exchange their current values of $t$ and then can immediately compute $m(e)$. This takes in total $O(\Delta)$ rounds. By Lemma~\ref{lem:local-payments} the payments can be computed with the same round complexity.
\end{proof}

\subsection{Distributed mechanism for dominating set} \label{sec:dominating-set}

In this section we design a fast distributed mechanism for computing a minimum weight dominating set. A \emph{dominating set} is a subset $D$ of nodes that covers all nodes, i.e.\ each node either has a neighbour in $D$ or is in $D$.  
Each selected agent in $D$ has valuation $-w(v)$ and each non-selected agent has valuation 0.

\begin{theorem} \label{thm:dominating-set-mechanism}
	There is an incentive-compatible distributed mechanism that computes a $(1+\ln (\Delta+1))$-approximation of minimum-weight dominating set and runs in time $O(\Delta^3 W + \log^* n)$ in the LOCAL model.
\end{theorem}

We use a distributed implementation of the greedy algorithm for approximating weighted set cover~\cite{johnson74approximation,chvatal79greedy}. In the weighted set cover problem, there is a universe $X$ of elements and a family $S$ of sets, each consisting of elements of $X$. Each set $s \in S$ has a weight $w(s)$, and the goal is to find a cover $C \subseteq S$ that 1) has $X$ as its union and 2) minimises the total weight of the cover.

Let $C$ be a partial cover, let $s$ be a set, and let $\bar{C}(s)$ be the set of uncovered elements $X \setminus \cup_{c \in C} c$. The \emph{ineffectiveness} of $s$ is defined as $f(s,C) = w(s)/|\bar{C}(s)|$. The sequential greedy algorithm repeatedly picks the element with the minimum ineffectiveness until all elements have been covered. It is known that if the sets have bounded size $k$, then the approximation factor of the algorithm is $H_k = 1 + 1/2 + \dots + 1/k \leq 1 + \ln k$ (see e.g.\ \cite[Chapter 2]{vazirani10approximation}).

Dominating set has a simple reduction to set cover: Let $G = (V,E,w)$ be a weighted graph. Set $X = V$. For each node $v \in V$ add the set $s(v) = N^+(v)$ to $S$ with $w'(s(v)) = w(s)$. Clearly a set cover $C$ of the instance $(U,S,w')$ corresponds to a dominating set of the graph $(V,E,w)$ with the same weight, and vice versa. When all sets have size at most $\Delta + 1$, the approximation factor is $1 + \ln (\Delta+1)$. For completeness, we give a proof for the distributed variant below, in Lemma~\ref{lem:set-cover-approx}.

The mechanism is implemented as a distributed algorithm as follows. Let $D_t$ denote the (partial) dominating set computed after $t$ steps of the loop. Initially we have $D_0 = \emptyset$. 
\begin{enumerate}[noitemsep]
	\item Compute a $(\Delta(G^2)+1)$-coloring $\phi$ of the $G^2$, the virtual graph obtained by connecting all pairs of nodes within distance 2 in $G$. This coloring will be used for tie breaking.
	\item Until $v$ becomes inactive, repeat the following steps. At step $t=1,2,\dots$, each node $v$ sends $f(v,D_{t-1})$ and $\phi(v)$ to its neighbours. Each node then picks the minimum value of $f$ among the received values and its own, breaking ties using $\phi$. It then sends the minimum value to its neighbours. If a node receives its own value from all neighbours, it is a local minimum and is selected to $D_{t}$. It announces joining $D$ to its neighbours and becomes inactive. Its neighbours $u$ update $f(u,D_{t})$ and announce to their neighbours if they became covered (so that they may update their ineffectiveness). If a node and all of its neighbours are covered, it becomes inactive. One step can be implemented in three communication rounds.
	\item Once a node is inactive and not in $D$, it computes its payment. 
\end{enumerate}

The structure of the algorithm differs from the first two examples: the quality of each node changes during execution as its degree decreases. To show that the algorithm is truthful, we first have to show that it is monotone.

\begin{lemma}[monotonicity] \label{lem:greedy-ds-mon}
	The greedy set cover algorithm is monotone. 
\end{lemma}

\begin{proof}
	To simplify the proof we initially analyse a different, non-adaptive variant of the algorithm. We will show that this algorithm is monotone and then show that it is equivalent to the original adaptive algorithm. 
	
	Let $K$ denote the number of values that the inefficiency $f$ can have: We have that $K = O(\Delta^3 W)$, as there are $\Delta W$ values from the function itself and $\Delta^2$ values from tie-breaking with the $O(\Delta^2)$-colouring. Let $x_1, x_2, \dots, x_K$ denote natural ordering of these values in the increasing order. 
	
	Now in each step $t$, a node $v$ joins the cover if and only if there is an uncovered node in $N^+(v)$ and $(f(v,D_{t-1}), \phi(v)) = x_t$. That is, nodes join at the round corresponding to their inefficiency (and tie-breaker).
	
	Now assume some node $v$ does not join $D$ with bid $b(v)$. We must show that it does not join with any bid $b'(v) > b(v)$ either. Call the two executions with bids $b(v)$ and $b'(v)$ \emph{original} and \emph{modified}, respectively. In the first step of the modified execution, the ineffectiveness of each node not $v$ is unchanged. Therefore the set of nodes that join in the first step are the same for both executions. Now assume that $D_{t} = D'_{t}$ (dominating sets for original and modified runs) up to some step $t$. Since the dominating sets agree up to step $t$, the ineffectivenesses also agree (except for $v$). Since $b'(v) > b(v)$ it follows that $v$ does not join since it did not join in the original execution. This implies that $v$ does not join in the modified run.
	
	Next we must show that the modified algorithm is equivalent to the original. For each $v \in D$ in the original algorithm, let $F(v)$ denote its inefficiency when it joins. Use induction on step $t = 1,2,\dots,K$ to show that exactly those nodes $v$ with $F(v) = x_t$ join in step $t$ in the modified algorithm. In the base case nodes with $F(v) = x_1$ join in the modified algorithm. Now assume that up to step $t$ exactly nodes $v$ with $F(v) = x_i$ join in each step $i \leq t$. Consider some $v$ with $F(v) = x_{t+1}$: Since $v$ is a local minimum, by induction assumption exactly those nodes that join before $v$ in its 2-neighbourhood in the original algorithm must have joined in the modified algorithm as well. Therefore inefficiency of $v$ is $F(v)$ in the modified algorithm as well and it joins. Conversely, assume some node $u$ joins in step $t$ in the modified algorithm. Since it is a local minimum (by definition) and the two algorithms agree on all lower inefficiency values, $u$ must join in the original algorithm as well with $F(u) = x_{t+1}$.
\end{proof}

Again, Lemma~\ref{lem:greedy-ds-mon} together with Theorem~\ref{thm:myerson} imply that there is a critical price that as the payment makes the mechanism incentive-compatible.

We finish this section by showing that the algorithm works correctly and bound its running time.

\begin{lemma} \label{lem:set-cover-approx}
	The greedy set-cover algorithm that selects all local minima computes a $(1+\ln (\Delta+1))$-approximation of minimum-weight dominating set.
\end{lemma}

\begin{proof}
	The proof follows from a standard argument (see e.g.\ \cite{vazirani10approximation}). We \emph{charge} each node $x$ for being covered by node $v$ at step $t$ with the cost $\alpha(x) = w(v)/|\bar{D}_t(v)|$ (i.e. spread the cost of $v$ equally among the new nodes it covers). The total cost of the solution computed by the greedy algorithm is $\sum_{x \in V} \alpha(x)$. 
	
	Let $D^* = \{v_1, v_2, \dots, v_k \}$ denote the nodes in an optimal solution. For each $N^+(v_i) = \{ x_1, x_2, \dots, x_{\ell} \}$, let the indexing denote the order in which they are covered by the greedy algorithm. At each step, at most one node is selected that covers nodes from $N^+(v_i)$ (select local 2-hop minima). When node $x_j$ is covered, the algorithm always has the option of choosing the node $v_i$: thus the charge on $x_j$ is at most $w(v_i)/(\ell-j+1)$. The total sum is at most
	\[
		\sum_{j=1}^{\ell} \alpha(j) \leq \sum_{j=1}^{\ell} w(v_i) / (\ell-j+1) = w(v_i) H_{\ell}, 
	\]
	where $H_n$ is the $n$th harmonic number.
	
	Note that $\ell \leq \Delta+1$, and the total cost of the greedy solution is at most
	\[
		\sum_{v_i \in D^*} \sum_{x_j \in N^+(v_i)} \alpha(x_j) \leq \sum_{v_i \in D^*} w(v_i) H_{\Delta+1} = H_{\Delta+1}w(D^*) \leq (1+\ln (\Delta+1))w(D^*).
	\]
	Last inequality follows from a standard bound on harmonic numbers.
\end{proof}

Finally, it remains to analyse the running time of the algorithm.

\begin{lemma}
	The mechanism can implemented in $O(\Delta^3 W + \log^* n)$ rounds in the LOCAL model.
\end{lemma}

\begin{proof}
	To see that the algorithm is correct, in each loop a node receives its own colour and its own value if it is the 2-hop local minimum. If it is not, the value it receives must have a different colour, as the colouring used for tie-breaking is a 2-hop colouring. Local minima are therefore selected correctly.
	
	Computing the 2-hop coloring of the input graph can be done using the $(\Delta+1)$-colouring algorith of Maus and Tonoyan~\cite{maus20linial} in $O(\Delta \log \Delta + \log^* n)$ rounds. Simulating the 2-hop network only requires constant overheard in the running time.
	
	There are $O(\Delta W)$ different values for $f$ and $O(\Delta^2)$ values for the tie-breaking for a total of $O(\Delta^3 W)$ values. In each step the current minimum value gets eliminated, and values can only increase during the execution. One step can be implemented in a constant number of communication rounds.
	
	By Lemma~\ref{lem:local-payments} the payments can be computed with the same time complexity.
\end{proof}

\subsection{Optimality}

The approximation factor is asymptotically optimal for computing a minimum-weight dominating set in $o(\log n)$ rounds with a distributed algorithm: There exist $d$-regular graphs with no dominating sets of size $o(n \log d / d)$ that are indistinguishable in $o(\log n)$ rounds from graphs with perfect dominating sets~(again, see \cite{alon10constant}, Lemma~2.1). This implies that no $o(\log n)$-time distributed algorithm can get better than a $\Omega(\log d)$-approximation on $d$-regular graphs. The dependence on $n$ is optimal, as breaking symmetry requires $\Omega(\log^* n)$ rounds~\cite{Linial1992}.

\subsection{Mechanism for slot assignment} \label{sec:slot-mechanism}

In this section we describe a mechanism for a weighted colouring-type problem, inspired by the centralised auctions for sponsored advertisements~\cite{edelman07internet,varian07position}. This example differs from the previous three in that the output is no longer binary. Instead, each agent is assigned a real value, indicating the value of the slot to which it is assigned.

\subsection{The slot assignment problem}

Each agent can be assigned to one of $s$ possible \emph{slots}. Each slot $j \in [s]$ has an intrinsic value called \emph{rate} $\alpha_j$ that is common to all agents. We assume that the rates are non-decreasing, i.e.\ for $j < j'$ we have $\alpha_j \leq \alpha_{j'}$. We will assume that $s = \Delta+1$, as settings with fewer slots can be modelled by setting $\alpha_j = 0$ for all $j > s$. Each agent $v$ values the rate with factor $w(v) \in \{1,\dots, B\}$, and assigning agent $v$ to slot $i$ gives valuation $\alpha_i w(v)$. 
The input graph is the network of conflicts: neighbouring agents cannot be assigned to the same slot.

We show the following theorem.
\begin{theorem} \label{thm:slot-auction-mechanism}
	There is an incentive-compatible distributed mechanism that computes an $O(\Delta)$-approximation of the slot assignment problem and runs in $O(\Delta B + \log^* n)$ rounds in the LOCAL model.
\end{theorem}

As far as we are aware, this problem has not been previously studied in the distributed setting. The original motivation for the slot auction is in selling slots for search engine advertisements. However, it can be generalised to the assignment of any interchangeable goods of deteriorating quality (with local conflicts). As an example, the colours could represent some schedule by which a locally conflicting operation is performed, and agents value being able to perform the operation as soon as possible.

\subsection{Algorithm}

We use a modified version of the algorithm for computing an independent set (Section~\ref{sec:independent-set}). The algorithm repeatedly picks the set of local maxima according to $<_{\phi}$ and assigns them to the best available slot. This can be done in parallel, as local maxima cannot be adjacent.

More formally, the algorithm consists of two phases, preprocessing and selection. Preprocessing occurs as described in Section~\ref{sec:independent-set}, where an order $<_{\phi}$ is computed. The selection phase is divided into steps. In the first round of the first step each node shares the bid of the corresponding agent to its neighbours. In the first rounds of corresponding steps each node sends its slot if got assigned one in the previous step. Each node keeps track of the slots used by its neighbours. Once a node receives messages in the first round of each step, it can determine whether it is a local maximum. If it is, it selects the largest available slot. In the next step it announces this selected slot and stops.

Payment function is the critical price given by Myerson's lemma. Assume focal agent $v$ bids $b(v)$ and is assigned to slot $c$. For a fixed $b_{-v}$, let $C(v) = \{c_1, c_2, \dots, c_{\ell} \}$ denote the set of possible slots (over all possible bids) for $v$. For $k \in \{1, \dots, \ell\}$, let $b^*_k(v)$ denote $k$th critical threshold of $v$: the smallest bid such that $v$ is assigned to slot $c_k$. The valuation of $v$ is $\alpha_c w(v)$, and its payment is $-\sum_{k=1}^c (\alpha_k - \alpha_{k-1}) b^*_k$ (where $\alpha_0 = 0$).

As we will show, this payment function has two key properties. First, the contribution of each slot is a constant independent of the bid of the focal agent -- the bid only determines which slots contribute. Second, the slots with a critical price above the weight of an agent have a non-positive contribution. This will ensure truthfulness.

\begin{lemma}
	The slot assignment mechanism is monotone and truthful.
\end{lemma}

This follows from Myerson's lemma and showing monotonicity, but we give an independent proof for completeness.

\begin{proof}
	In each step of the algorithm the current local maxima are selected and assigned lowest available slots. Assume that with bid $b(v)$ agent $v$ is assigned to a slot in step $t$ of the algorithm. Now consider the alternative execution, where $v$ bids $b'(v) > b(v)$ instead: in each step $t' < t$ either the set of local maxima is the same as in the original execution, or $v$ is a local maximum. In the first case the two executions do not differ, and in the second case $v$ is assigned a slot in step $t'$. Since the set of used slots in its neighbourhood is a subset of those that are used at step $t$ with bid $b(v)$, it is assigned to a slot $j' \geq j$. If $i$ is not selected in any step $t'$ then the two executions agree up to step $t$ and $i$ is a local maximum with bid $b'(v)$. Therefore it is selected and assigned the same slot as in the original execution, establishing the monotonicity. 
	
 	Next prove that overbidding (or underbidding) does not improve utility. Assume that the true valuation $w(v)$ of agent $v$ is such that it is assigned to the $j$th possible slot with bid $b(v) = w(v)$, i.e.\ we have that $b^*_j(v) \leq w(v) < b^*_{j+1}(v)$. The truth has valuation $\alpha_{c_j} w(v) - \sum_{k=1}^j (\alpha_{c_k} - \alpha_{c_{k-1}}) b^*_k$. By definition, bidding $b'(v): b^*_{j}(v) \leq b'(v) < b^*_{j+1}$ has no effect on the utility. Bidding $b'(v) \geq b^*_{j+1}$ causes $v$ to be assigned some slot $c_{j'} \geq c_{j+1}$. The valuation of $v$ is now 
	\begin{align*}
			u_v(b', c_{j'}) &= \alpha_{c_{j'}} w(v) - \sum_{k=1}^{j'} (\alpha_{c_k} - \alpha_{c_{k-1}}) b^*_k(v) \\
			&= u_v(b, c_{j}) + (\alpha_{c_{j'}} - \alpha_{c_j}) w(v) - \sum_{k=j+1}^{j'} (\alpha_{c_k} - \alpha_{c_{k-1}}) b^*_k(v) < u_v(b,c_j).
	\end{align*}
	The last inequality follows from the fact that the true valuation $w(v)$ is less than critical thresholds $b^*_{j+1}, \dots b^*_{j'}$. By a similar argument a agent should not underbid, as each slot that can be gained with a truthful bid gives non-negative utility.
	
	Since the critical thresholds are independent of the bid of the focal agent $v$ and increasing, it follows that bidding $b(v) = w(v)$ maximises utility.
\end{proof}

\begin{lemma}
	The mechanism can be implemented with a total running time $O(\log^* n + B\Delta)$, where $B$ denotes the maximum rate of an agent, in the LOCAL model.
\end{lemma}

\begin{proof}
	The preprocessing step of colouring the graph with $(\Delta+1)$ colours can be done in $O(\log^* n + B\Delta)$ rounds~\cite{maus20linial}.
	
	In each step of the second phase the local maxima are assigned slots and stop. The ordering $<_{\phi}$ has in total $(\Delta+1)B$ values, so after that many steps all nodes have stopped. One step can be implemented in 2 rounds in the LOCAL model. Critical prices can be computed based on Lemma~\ref{thm:myerson}.
\end{proof}

\begin{lemma}
	The greedy algorithm has approximation ratio at most $\alpha_1 (\Delta+1)/(\sum_{i=1}^{\Delta+1} \alpha_i)$.
\end{lemma}

\begin{proof}
	Since the slots are assigned greedily, each agent assigned to slot $k$ has neighbours of larger weight in slots $k+1, k+2, \dots, \Delta+1$.
	
	This means that each agent assigned to slot $k$ has at most $\Delta - k$ neighbours at a lower level -- neighbours that it can block from taking slot $k$. Let $S_1, S_2, \dots, S_{\Delta+1}$ denote the sets of nodes assigned to slots $1, 2, \dots, \Delta+1$, respectively, by the greedy algorithm. We will analyse the worst-case weight distribution of the output. 
	
	Since each agent in slot 2 has at most one neighbour in slot 1 (it requires $\Delta-1$ edges to higher slots), we have that $w(S_1) \leq w(S_2)$. Each agent in $S_3$ has at most two neighbours in lower slots, and its weight has to be higher than both (worst case): we get that $w(S_3) \geq (w(S_2) + w(S_1))/2$. More generally, it is easy to see that $w(S_k)$ is at least the average weight of the lower slots: $w(S_k) \geq (\sum_{i=1}^{k-1} w(S_i))/(k-1)$.
	
	We have $w(S_{\Delta+1}) \geq (w(V) - w(S_{\Delta+1}))/\Delta$, implying that $w(S_{\Delta+1}) \geq w(V)/(\Delta+1)$. More generally, for any $j \leq \Delta$ we have that $\sum_{i=0}^{j-1} w(S_{\Delta+1-i}) \geq jw(V)/(\Delta+1)$. Again, assuming worst case, we have that the total valuation of the greedy solution $o$ is lower bounded by 
	\[
		u(o) \geq \biggl(\sum_{i=1}^{\Delta+1} \alpha_i\biggr)w(V)/(\Delta+1).
	\]
	Now let $o^*$ denote the optimum solution. We have that
	\[
	\frac{w(o^*)}{w(o)} \leq \frac{\alpha_1 (\Delta+1)}{\sum_{i=1}^{\Delta+1} \alpha_i}.
	\]
\end{proof}

\section{Discretisation} \label{sec:discretisation}

In this section we show how to deal with real-valued parameters. The mechanisms, as presented (except the local ratio algorithm for vertex cover) do not work well in this setting, as real-valued parameters allow for arbitrarily long dependency chains, making distributed implementation inefficient. We deal with this issue by discretisation: subdividing the interval of possible parameter values into smaller intervals, each of which is considered a single value in the mechanism. We show that truthfulness is retained while the quality of the solution is affected as a function of the size of a discretisation step.

\subsection{Discretisation for binary optimisation problems}

So far we assumed that the weights came from the set $\{0, 1, 2, \dots, W\}$ for some maximum weight parameter $W$. Now we generalise this so that the weights of the agents come from the real interval $[0,W]$. We assume that the underlying algorithm functions correctly for any set of $W$ distinct weights such that the running time dependency on the weights is a function of $W$. We call a mechanism based on such an algorithm \emph{scale-free}.

We model the discretisation as follows. Each agent $v$ now has a weight $w(v) \in [0,W]$. The mechanism designer fixes a parameter $\varepsilon > 0$: this is the size of the discretisation step. For convenience we assume that $K = W/\varepsilon$ is an integer. The set of possible values used by the discretised mechanism is defined as $W_{\varepsilon} = \{0, \varepsilon, 2\varepsilon, \dots, K\varepsilon \}$. When an agent submits its real-valued bid $b(v)$, if the problem is a maximisation problem this is converted into a discrete value by rounding down in $W_{\varepsilon}$: by choosing the maximum value $s \in W_{\varepsilon} : s \leq b(v)$. Denote the bid vector obtained this way by $b_{\varepsilon}$. The direction of the rounding is important for the truthfulness of our mechanisms: for minimisation problems, the rounding is up in $W_{\varepsilon}$ instead.

After this discretisation step, a mechanism designed with the assumption of discrete values can be run with the bids given by $b_{\varepsilon}$. 

\begin{theorem}
	Assume that $M$ is a scale-free truthful mechanism with a monotone assignment rule for a binary optimisation problem. Then $M_{\varepsilon}$, the same mechanism with the added discretisation step, is also truthful.
\end{theorem}

\begin{proof}
	Let $\lfloor w(v) \rfloor$ and $\lceil w(v) \rceil$ denote the closest lower and higher values, respectively, to $w(v)$ in $W_{\varepsilon}$. Observe that that the critical price is by definition in $W_{\varepsilon}$.
	
	First, consider binary maximisation problems. Let $S$ denote the set selected by the algorithm. Assume that a focal agent $v$ bidding $b(v) = w(v)$ results in $v \notin S$. This implies that $b^*(v) \geq \lceil w(v) \rceil$. Bidding $b(v) \geq b^*(v)$ results in non-positive utility, and all bids $b(v) < b^*(v)$ result in utility 0.
	Now assume that truthful bidding results in $v \in S$. Since the discretisation step is rounded down, this implies that $w(v) \geq b^*(v)$. All bids above $b^*(v)$ give the same non-negative utility, and all bids $b(v) < b^*(v)$ give utility 0.
	
	An analogous result applies for binary minimisation problems, where rounding is in the opposite direction. If focal agent $v$ is selected with the truth, we have that $b^*(v) \geq \lceil w(v) \rceil$. Total utility is the same non-negative value $b^*(v) - w(v)$ for all bids $b(v) \leq b^*(v)$ and 0 otherwise. If focal agent is not selected with the truth we have that $b^*(v) < w(v)$. Total utility for any $b(v) > b^*(v)$ is 0 and total utility for any $b(v) \leq b^*(v)$ is non-positive $b^*(v) - w(v)$.
\end{proof}

\subsection{Solution quality and running time}

The size of the discretisation step affects the quality of the solution and the number of different weights (and therefore running time). The exception to the latter is the local ratio algorithm for vertex cover from Section~\ref{sec:vertex-cover-mechanism}: its running time has no dependency on the maximum weight.

\begin{theorem}
	Assume that a mechanism $M$ computes an $\alpha$-approximation to binary optimisation problem $P$. Then the discretised mechanism $M_{\varepsilon}$ computes an $\alpha/(1-\varepsilon)$-approximation to $P$.
\end{theorem}

\begin{proof}
	Assume $P$ is a maximisation problem. The proof for minimisation problems is analogous. Since we are assuming that the agents are follow their (weakly) dominant strategy, let $w_{\varepsilon}$ denote the discretised weight function and use this directly instead of the bids. Since the weight of each node changes in the rounding by at most $\varepsilon$, we have that  $w_{\varepsilon}(S) \geq (1-\varepsilon)w(S)$ for any solution $S$. This holds in particular for any optimal solution under $w_{\varepsilon}$, denoted by $S_{\varepsilon}^*$. Therefore, considering any optimal solution $S^*$, we have $w_{\varepsilon}(S_{\varepsilon}^*) \geq w_{\varepsilon}(S^*) \geq (1-\varepsilon)w(S^*)$. 
	The mechanism $M_{\varepsilon}$ computes an $\alpha$-approximation $S$ with respect to $w_{\varepsilon}$, and therefore we have that 
	\[
	w(S) \geq w_{\varepsilon}(S) \geq w_{\varepsilon}(S^*_{\varepsilon}) / \alpha \geq (1-\varepsilon)w(S^*) / \alpha.
	\] 
\end{proof}

\section*{Acknowledgements} We thank Dennis Olivetti for pointing out the connection between colourings and orientations of bounded length. We thank Daniel Hauser, Mitri Kitti, Pauli Murto, Jukka Suomela, and Juuso V{\"a}lim{\"a}ki for discussions.

\bibliographystyle{alphaurl}
\bibliography{bibliography}

@article{nash50equilibrium,
author = {John F. Nash },
title = {Equilibrium points in $n$-person games},
journal = {Proceedings of the National Academy of Sciences},
volume = {36},
number = {1},
pages = {48-49},
year = {1950},
doi = {10.1073/pnas.36.1.48}
}

@article{jordehi19optimisation,
title = {Optimisation of demand response in electric power systems, a review},
journal = {Renewable and Sustainable Energy Reviews},
volume = {103},
pages = {308-319},
year = {2019},
issn = {1364-0321},
doi = {10.1016/j.rser.2018.12.054},
author = {A. Rezaee Jordehi}
}

@article{varian07position,
title = {Position auctions},
journal = {International Journal of Industrial Organization},
volume = {25},
number = {6},
pages = {1163--1178},
year = {2007},
doi = {10.1016/j.ijindorg.2006.10.002},
author = {Hal R. Varian}
}

@article{clarke71,
title={Multipart pricing of goods},
journal={Public Choice},
volume={11},
pages = {17--33},
year={1971},
author = {Edward H. Clarke}}

@article{groves73,
author = {Theodore Groves},
title = {Incentives in Teams},
journal = {Econometrica},
volume = {41},
pages = {617--631},
year = {1973}}

@article{calinescu15bounding,
title = {Bounding the payment of approximate truthful mechanisms},
journal = {Theoretical Computer Science},
volume = {562},
pages = {419--435},
year = {2015},
doi = {10.1016/j.tcs.2014.10.019},
author = {Gruia Calinescu}
}

@inbook{feigenbaum07book, 
place={Cambridge}, 
title={Distributed Algorithmic Mechanism Design}, 
booktitle={Algorithmic Game Theory}, 
publisher={Cambridge University Press}, 
author={Feigenbaum, Joan and Schapira, Michael and Shenker, Scott},
year={2007}, 
pages={363–-384}}

@article{carroll11distributed,
title = {Distributed algorithmic mechanism design for scheduling on unrelated machines},
journal = {Journal of Parallel and Distributed Computing},
volume = {71},
number = {3},
pages = {397--406},
year = {2011},
doi = {10.1016/j.jpdc.2010.11.004},
author = {Thomas E. Carroll and Daniel Grosu}
}

@article{feigenbaum01sharing,
title = {Sharing the Cost of Multicast Transmissions},
journal = {Journal of Computer and System Sciences},
volume = {63},
number = {1},
pages = {21--41},
year = {2001},
doi = {10.1006/jcss.2001.1754},
author = {Joan Feigenbaum and Christos H. Papadimitriou and Scott Shenker}
}

@inproceedings{feigenbaum02bgp,
author = {Feigenbaum, Joan and Papadimitriou, Christos and Sami, Rahul and Shenker, Scott},
title = {A BGP-based mechanism for lowest-cost routing},
year = {2002},
publisher = {Association for Computing Machinery},
doi = {10.1145/571825.571856},
booktitle = {Proc. 21st Annual Symposium on Principles of Distributed Computing (PODC)},
pages = {173–-182}
}

@article{li10mechanism,
title = {Mechanism design for set cover games with selfish element agents},
journal = {Theoretical Computer Science},
volume = {411},
number = {1},
pages = {174--187},
year = {2010},
doi = {10.1016/j.tcs.2009.09.024},
author = {Xiang-Yang Li and Zheng Sun and WeiZhao Wang and Xiaowen Chu and ShaoJie Tang and Ping Xu}
}

@inproceedings{devanur03strategy,
author = {Devanur, Nikhil R. and Mihail, Milena and Vazirani, Vijay V.},
title = {Strategyproof cost-sharing mechanisms for set cover and facility location games},
year = {2003},
publisher = {ACM},
doi = {10.1145/779928.779942},
booktitle = {Proc. 4th ACM Conference on Electronic Commerce},
pages = {108–-114}
}

@book{vazirani10approximation,
author = {Vazirani, Vijay V.},
title = {Approximation Algorithms},
year = {2010},
isbn = {3642084699},
publisher = {Springer Publishing Company, Incorporated}
}

@article{baryehuda04local,
author = {Bar-Yehuda, Reuven and Bendel, Keren and Freund, Ari and Rawitz, Dror},
title = {Local ratio: A unified framework for approximation algorithms. In Memoriam: Shimon Even 1935-2004},
year = {2004},
issue_date = {December 2004},
publisher = {Association for Computing Machinery},
volume = {36},
number = {4},
doi = {10.1145/1041680.1041683},
journal = {ACM Comput. Surv.},
pages = {422–-463}
}

@article{baryehuda81linear,
title = {A linear-time approximation algorithm for the weighted vertex cover problem},
journal = {Journal of Algorithms},
volume = {2},
number = {2},
pages = {198--203},
year = {1981},
doi = {10.1016/0196-6774(81)90020-1},
author = {Reuven Bar-Yehuda and Shimon Even}
}

@InProceedings{hirvonen24fast,
  author       = {Juho Hirvonen and
                  Sara Ranjbaran},
  title        = {Fast, Fair and Truthful Distributed Stable Matching for Common Preferences},
  booktitle = {Proc. 28th International Conference on Principles of Distributed Systems (OPODIS 2024)},
  year         = {2024},
  series =	{Leibniz International Proceedings in Informatics (LIPIcs)},
  url          = {https://doi.org/10.48550/arXiv.2402.16532},
  note         = {To appear.}
}

@inproceedings{ostrovsky2015fast,
  author       = {Rafail Ostrovsky and
                  Will Rosenbaum},
  title        = {Fast Distributed Almost Stable Matchings},
  booktitle    = {Proc. 2015 {ACM} Symposium on Principles of Distributed
                  Computing ({PODC} 2015)},
  pages        = {101--108},
  publisher    = {{ACM}},
  year         = {2015},
  doi          = {10.1145/2767386.2767424},
  bibsource    = {dblp computer science bibliography, https://dblp.org}
}

@InProceedings{hirvonen23convergence,
  author =	{Hirvonen, Juho and Schmid, Laura and Chatterjee, Krishnendu and Schmid, Stefan},
  title =	{{On the Convergence Time in Graphical Games: A Locality-Sensitive Approach}},
  booktitle =	{Proc. 27th International Conference on Principles of Distributed Systems (OPODIS 2023)},
  pages =	{11:1--11:24},
  series =	{Leibniz International Proceedings in Informatics (LIPIcs)},
  year =	{2024},
  volume =	{286},
  publisher =	{Schloss Dagstuhl -- Leibniz-Zentrum f{\"u}r Informatik},
  address =	{Dagstuhl, Germany},
  doi =		{10.4230/LIPIcs.OPODIS.2023.11}
}

@inproceedings{elkind07frugality,
author = {Elkind, Edith and Goldberg, Leslie Ann and Goldberg, Paul W.},
title = {Frugality ratios and improved truthful mechanisms for vertex cover},
year = {2007},
publisher = {Association for Computing Machinery},
doi = {10.1145/1250910.1250959},
booktitle = {Proc. 8th ACM Conference on Electronic Commerce},
pages = {336–-345}
}

@inproceedings{afek14building,
author = {Afek, Yehuda and Ginzberg, Yehonatan and Landau Feibish, Shir and Sulamy, Moshe},
title = {Distributed computing building blocks for rational agents},
year = {2014},
publisher = {Association for Computing Machinery},
doi = {10.1145/2611462.2611481},
booktitle = {Proc. 2014 ACM Symposium on Principles of Distributed Computing (PODC)},
pages = {406-–415}
}

@misc{nobel12,
  title={Scientific Background on the Sveriges Riksbank Prize in Economic Sciences in Memory of Alfred Nobel 2012: Stable allocations and the practice of market design},
  author={Economic Sciences Prize Committee of the Royal Swedish Academy of Sciences},
  howpublished={\url{https://www.kva.se/app/uploads/2012/10/globalassets-priser-ekonomi-2012-scibackeken12.pdf}},
  note = {Accessed:15-02-2024},
  year = {2012}
}

@article{roth04kidney,
    author = {Roth, Alvin E. and S{\"o}nmez, Tayfun and {\"U}nver, M. Utku},
    title = "{Kidney Exchange}",
    journal = {The Quarterly Journal of Economics},
    volume = {119},
    number = {2},
    pages = {457-488},
    year = {2004},
    month = {05},
    doi = {10.1162/0033553041382157}
}

@article{edelman07internet,
Author = {Edelman, Benjamin and Ostrovsky, Michael and Schwarz, Michael},
Title = {Internet Advertising and the Generalized Second-Price Auction: Selling Billions of Dollars Worth of Keywords},
Journal = {American Economic Review},
Volume = {97},
Number = {1},
Year = {2007},
Pages = {242–-259},
DOI = {10.1257/aer.97.1.242}}

@article{vickrey61auction,
author = {Vickrey, William},
title = {COUNTERSPECULATION, AUCTIONS, AND COMPETITIVE SEALED TENDERS},
journal = {The Journal of Finance},
volume = {16},
number = {1},
pages = {8--37},
doi = {10.1111/j.1540-6261.1961.tb02789.x},
year = {1961}
}

@article{myerson81optimal,
author = {Myerson, Roger B.},
title = {Optimal Auction Design},
journal = {Mathematics of Operations Research},
volume = {6},
number = {1},
pages = {58--73},
year = {1981},
doi = {10.1287/moor.6.1.58},
}

@article{Naor1995,
  author = {Naor, Moni and Stockmeyer, Larry},
  doi = {10.1137/S0097539793254571},
  journal = {SIAM Journal on Computing},
  number = {6},
  pages = {1259--1277},
  title = {{What Can be Computed Locally?}},
  volume = {24},
  year = {1995}
}

@inproceedings{maus20linial,
  author       = {Yannic Maus and
                  Tigran Tonoyan},
  editor       = {Hagit Attiya},
  title        = {Local Conflict Coloring Revisited: Linial for Lists},
  booktitle    = {Proc. 34th International Symposium on Distributed Computing ({DISC} 2020)},
  series       = {LIPIcs},
  volume       = {179},
  pages        = {16:1--16:18},
  publisher    = {Schloss Dagstuhl - Leibniz-Zentrum f{\"{u}}r Informatik},
  year         = {2020},
  doi          = {10.4230/LIPICS.DISC.2020.16},
  bibsource    = {dblp computer science bibliography, https://dblp.org}
}

@article{sakai03note,
title = {A note on greedy algorithms for the maximum weighted independent set problem},
journal = {Discrete Applied Mathematics},
volume = {126},
number = {2},
pages = {313-322},
year = {2003},
issn = {0166-218X},
doi = {10.1016/S0166-218X(02)00205-6},
author = {Shuichi Sakai and Mitsunori Togasaki and Koichi Yamazaki}
}

@article{chvatal79greedy,
 author = {V. Chvatal},
 journal = {Mathematics of Operations Research},
 number = {3},
 pages = {233--235},
 publisher = {INFORMS},
 title = {A Greedy Heuristic for the Set-Covering Problem},
 urldate = {2024-04-12},
 volume = {4},
 year = {1979},
 url = {http://www.jstor.org/stable/3689577}
}

@article{johnson74approximation,
title = {Approximation algorithms for combinatorial problems},
journal = {Journal of Computer and System Sciences},
volume = {9},
number = {3},
pages = {256--278},
year = {1974},
doi = {10.1016/S0022-0000(74)80044-9},
author = {David S. Johnson}
}

@article{roughgarden2007routing,
  title={Routing games},
  author={Roughgarden, Tim},
  year={2007},
  journal={Algorithmic game theory},
  volume={18},
  pages={459--484},
  publisher={Cambridge University Press Cambridge, MA}
}

@article{Linial1992,
  author = {Linial, Nathan},
  doi = {10.1137/0221015},
  journal = {SIAM Journal on Computing},
  number = {1},
  pages = {193--201},
  title = {{Locality in Distributed Graph Algorithms}},
  volume = {21},
  year = {1992}
}

@book{Peleg2000,
  author = {Peleg, David},
  doi = {10.1137/1.9780898719772},
  publisher = {Society for Industrial and Applied Mathematics},
  title = {{Distributed Computing: A Locality-Sensitive Approach}},
  year = {2000}
}

@article{Luby1986,
  author = {Luby, Michael},
  doi = {10.1137/0215074},
  journal = {SIAM Journal on Computing},
  number = {4},
  pages = {1036--1053},
  title = {{A Simple Parallel Algorithm for the Maximal Independent Set Problem}},
  volume = {15},
  year = {1986}
}

@article{floreen10stable,
  author    = {Patrik Flor{\'{e}}en and
               Petteri Kaski and
               Valentin Polishchuk and
               Jukka Suomela},
  title     = {Almost Stable Matchings by Truncating the Gale-Shapley Algorithm},
  journal   = {Algorithmica},
  volume    = {58},
  number    = {1},
  pages     = {102--118},
  year      = {2010},
  doi       = {10.1007/s00453-009-9353-9},
  bibsource = {dblp computer science bibliography, https://dblp.org}
}

@book{nisan07algorithmic,
  place={Cambridge},
  title={Algorithmic Game Theory},
  DOI={10.1017/CBO9780511800481},
  publisher={Cambridge University Press},
  year={2007},
  editor={Noam Nisan and Tim Roughgarden and Eva Tardos and Vijay V. Vazirani}
}

@inproceedings{collet18equilibria,
  author    = {Simon Collet and
               Pierre Fraigniaud and
               Paolo Penna},
  title     = {Equilibria of Games in Networks for Local Tasks},
  booktitle = {Proc. 22nd International Conference on Principles of Distributed Systems
               ({OPODIS} 2018)},
  series    = {LIPIcs},
  volume    = {125},
  pages     = {6:1--6:16},
  publisher = {Schloss Dagstuhl - Leibniz-Zentrum f{\"{u}}r Informatik},
  year      = {2018},
  doi       = {10.4230/LIPIcs.OPODIS.2018.6},
  bibsource = {dblp computer science bibliography, https://dblp.org}
}

@inbook{alon10constant,
author = {Alon, Noga},
title = {On Constant Time Approximation of Parameters of Bounded Degree Graphs},
year = {2010},
publisher = {Springer-Verlag},
address = {Berlin, Heidelberg},
booktitle = {Property Testing: Current Research and Surveys},
pages = {234–239},
numpages = {6}
}

@article{gale62sm,
author = {Gale, David and Shapley, Lloyd S.},
title = {College Admissions and the Stability of Marriage},
journal = {The American Mathematical Monthly},
volume = {69},
number = {1},
year = {1962},
pages = {9--15},
doi = {10.2307/2312726}
}

@article{Acher-approximation,
  author       = {Aaron Archer and
                  Joan Feigenbaum and
                  Arvind Krishnamurthy and
                  Rahul Sami and
                  Scott Shenker},
  title        = {Approximation and collusion in multicast cost sharing},
  journal      = {Games Econ. Behav.},
  volume       = {47},
  number       = {1},
  pages        = {36--71},
  year         = {2004},
  doi          = {10.1016/S0899-8256(03)00176-3}
}

@article{Prabodini-IoT,
  author       = {Prabodini Semasinghe and
                  Setareh Maghsudi and
                  Ekram Hossain},
  title        = {Game Theoretic Mechanisms for Resource Management in Massive Wireless
                  IoT Systems},
  journal      = {{IEEE} Commun. Mag.},
  volume       = {55},
  number       = {2},
  pages        = {121--127},
  year         = {2017},
  doi          = {10.1109/MCOM.2017.1600568CM},
  bibsource    = {dblp computer science bibliography, https://dblp.org}
}

\appendix

\section{Proof of Lemma~\ref{lem:is-approx}} \label{app:is-apx-proof}

\begin{proof}
	The algorithm proceeds in steps, each of which consists of adding set of local maxima (according to $<$) to $I$ and then removing them and their neighbours from the graph. Let $G_i$ denote the graph after the $i$th step ($G_0 = G$). Similarly, let $I_i$ denote the nodes added to the independent set in step $i$. Finally, let $\alpha(G)$ denote the weight of the optimal solution in $G$.
	
	First observe that for each $v \in I_i$ it holds that $w(v) \geq (\sum_{u \in N_{G_i}(v)} w(u)) / \deg_{G_i}(v)$: the weight of a local maximum is at least the average weight of its neighbours active when it is picked. Next, observe that for each $v \in I_i$, we have that $\alpha(G[N^+_{G_i}(v)]) \leq \max \{ w(v), \sum_{u \in N_{G_i}(v)} w(u) \}$: the optimum solution constrained to a neighbourhood is either the focal node or all of its neighbours. In both cases it holds that $\alpha(G[N^+_{G_i}(v)]) \leq \Delta w(v)$. 
	
	Since all nodes are processed by the end, we have that 
	\[
	\alpha(G) \leq \sum_{i = 1}^T \sum_{v \in I_i} \alpha(G[N^+_{G_i}(v)]) \leq \Delta w(I). \qedhere
	\]
	This first inequality follows from the fact that when each node is picked, its weight is at most the weight of the heaviest independent set in the subgraph induced by it and its active neighbours.
\end{proof}

\end{document}